\documentclass[a4paper,11pt]{article}

\usepackage[latin1]{inputenc}

\usepackage{indentfirst} % pour indenter automatiquement le 1er paragraphe d'une section

\usepackage[english]{babel}
\usepackage{graphicx}
\usepackage{subfigure}
\usepackage{amssymb}

\usepackage{eqnarray}
\usepackage{array}
\usepackage{amsmath, amsthm}

\title{\large \textbf{AN FPT ALGORITHM FOR PLANAR MULTICUTS WITH SOURCES AND SINKS ON THE OUTER FACE}}

\author{Cédric Bentz\footnote{CNAM \& CEDRIC Laboratory, 292 rue Saint Martin, 75003 Paris (France). \protect \\ Phone: +33 (0) 1 58 80 86 14. E-mail address:
cedric.bentz@cnam.fr}}

\date{}

\newtheorem{Lemma}{Lemma}
\newtheorem{Corollary}{Corollary}
\newtheorem{Th}{Theorem}

\begin{document}

\maketitle

\begin{abstract}
Given a list of $k$ source-sink pairs in an edge-weighted graph $G$, the \emph{minimum multicut problem} consists in selecting a set of edges of minimum total weight in $G$, such that removing these edges leaves no path from each source to its corresponding sink. To the best of our knowledge, no non-trivial FPT result for special cases of this problem, which is \textbf{APX}-hard in general graphs for any fixed $k \geq 3$, is known with respect to $k$ only. When the graph $G$ is planar, this problem is known to be polynomial-time solvable if $k = O(1)$, but cannot be FPT with respect to $k$ under the \emph{Exponential Time Hypothesis}.

In this paper, we show that, if $G$ is planar and in addition all sources and sinks lie on the outer face, then this problem does admit an FPT algorithm when parameterized by $k$ (although it remains \textbf{APX}-hard when $k$ is part of the input, even in stars). To do this, we provide a new characterization of optimal solutions in this case, and then use it to design a ``divide-and-conquer'' approach: namely, some edges that are part of any such solution actually define an optimal solution for a polynomial-time solvable multiterminal variant of the problem on some of the sources and sinks (which can be identified thanks to a reduced enumeration phase). Removing these edges from the graph cuts it into several smaller instances, which can then be solved recursively.
\end{abstract}

\section{Introduction}\label{sectIntro}

Given a list of $k$ pairs (source $s_i$, sink $s'_i$) in an undirected edge-weighted graph $G$, the \emph{minimum multicut problem} (\textsc{MinMC}) consists in selecting a set of edges of minimum total weight in $G$, in such a way that removing these edges leaves no path between $s_i$ and $s'_i$ for each $i$. As the weight $w(e)$ of each edge $e$ is commonly assumed to be a rational number, we can actually assume without loss of generality that each $w(e)$ is an integer (by multiplying all $w(e)$'s by one sufficiently large integer).

A well-known special case of \textsc{MinMC} is the \emph{minimum multiway cut problem}, or \emph{minimum multiterminal cut problem} (\textsc{MinMTC}): in any instance of this problem, we are given a set of \emph{terminals} $\mathcal{T} = \{t_1, \dots, t_{\arrowvert \mathcal{T} \arrowvert}\}$, and the source-sink pairs in the associated \textsc{MinMC} instance are $(t_i, t_j)$ for $i \neq j$.

We shall only consider undirected graphs here. When $k=1$, \textsc{MinMC} (which is then equivalent to \textsc{MinMTC} with $\vert \mathcal{T} \vert = 2$) turns into the famous {\em minimum cut problem}, and can therefore be solved in polynomial time. Moreover, \textsc{MinMC} remains polynomial-time solvable when $k=2$~\cite{refYan83}. However, \textsc{MinMTC} is \textbf{APX}-hard for any fixed value of $\vert \mathcal{T} \vert \geq 3$~\cite{refDahlhaus94}: note that, when $\vert \mathcal{T} \vert = 3$, \textsc{MinMTC} is actually a special case of \textsc{MinMC} with $k = 3$.

When $k$ is part of the input, \textsc{MinMC} is tractable in chains, but \textbf{APX}-hard even in stars with weights 1~\cite{refGVY97}, and hence also in planar graphs where all sources and sinks lie on the outer face. However, when $k = O(1)$, it becomes tractable in trees, and even in graphs of bounded tree-width~\cite{refBentzDAM08}.

Some results are known about the parameterized complexity of \textsc{MinMC}. For instance, it is known to be FPT with respect to the solution size~\cite{refBDT11,refMR14}, but we are not aware of any non-trivial FPT result when the parameter to be considered is $k$. Recall that a problem parameterized by some parameter $p$ is FPT with respect to $p$ if it admits an FPT algorithm with respect to $p$, i.e., an algorithm solving it in time $O(f(p) n^c)$, where $f(\cdot)$ is some computable function of $p$, $n$ is the input size, and $c$ is a constant independent of $p$ \cite{refDowney99}.

Let us now turn to the case where $G$ is planar. On the one hand, when all sources and sinks lie on the outer face, it was proved that, unlike \textsc{MinMC}, \textsc{MinMTC} can be solved in polynomial time, even when $\vert \mathcal{T} \vert$ is part of the input~\cite{refChen04}. On the other hand, when sources and sinks can lie anywhere, it was proved in~\cite{refMarx12} that, under the \emph{Exponential Time Hypothesis} (ETH), \textsc{MinMTC} cannot be FPT with respect to $\vert \mathcal{T} \vert$ in planar graphs. Hence, under the same hypothesis, \textsc{MinMC} cannot be FPT with respect to $k$ in these graphs. However, when $k = O(1)$, it was proved that \textsc{MinMC} is polynomial-time solvable in planar graphs if all sources and sinks lie on the outer face~\cite{refBentzDAM09}, and later it was proved that \textsc{MinMC} remains polynomial-time solvable in planar graphs even when sources and sinks can lie anywhere~\cite{refBentzIPEC12,refECDV17}. (It was already known for \textsc{MinMTC} in planar graphs when $\vert \mathcal{T} \vert = O(1)$~\cite{refDahlhaus94}.)

In the present paper, we prove the first non-trivial FPT result concerning \textsc{MinMC} parameterized by $k$ only, and at the same time settle the last case left open by the results shown in~\cite{refBentzDAM09,refBentzIPEC12,refECDV17,refMarx12}. Namely, we show that \textsc{MinMC} is FPT with respect to $k$ when $G$ is planar and all sources and sinks lie on its outer face. In order to do this, we provide a new characterization of optimal solutions for \textsc{MinMC} in such graphs, on which our algorithm is based.

In~\cite{refBentzDAM09}, it was proved that any \textsc{MinMC} instance in such a graph can be reduced to a set of \textsc{MinMTC} instances \emph{in planar graphs} (i.e., where sources and sinks can lie anywhere). Actually, a limited number of configurations were enumerated, and for each configuration \emph{one} planar \textsc{MinMTC} instance was solved (using a non-FPT algorithm, such as the one in~\cite{refDahlhaus94}).

Here, we prove a stronger result: there exist optimal multicuts such that some part (i.e., some of the edges) of such a solution actually defines an optimal solution for a \textsc{MinMTC} instance, obtained in the same graph by removing some of the sources and sinks (or, equivalently, by keeping only some of them). (In practice, determining the sources and sinks that belong to this \textsc{MinMTC} instance requires some enumeration, but fortunately it can be done in FPT time.) In the \textsc{MinMTC} instance obtained in this way, \emph{all} terminals lie on the outer face (and hence we can use the polynomial-time algorithm given in~\cite{refChen04}). Moreover, removing the edges of the optimal solution for this \textsc{MinMTC} instance cuts the initial graph into several pieces (i.e., connected components), which can then be solved recursively as smaller \textsc{MinMC} instances satisfying the same assumptions as the initial one.

The proposed algorithm is thus based on a \emph{divide-and-conquer} approach: we enumerate a limited (but larger than in~\cite{refBentzDAM09}, as we will need to ``guess'' slightly more information) number of configurations, and for each one we solve \emph{a set} of planar \textsc{MinMTC} instances (and not \emph{one} planar \textsc{MinMTC} instance anymore), in which, unlike in~\cite{refBentzDAM09}, \emph{all} terminals lie on the outer face. This enables us to obtain a nearly linear-time algorithm when $k=O(1)$.

When considering any planar \textsc{MinMC} (or \textsc{MinMTC}) instance, loops and parallel edges are useless (loops can be removed, and edges having the same endpoints can be merged into a single edge, whose weight is the sum of the weights of the merged edges), and connected components actually define independent instances, so we shall assume without loss of generality that the input graph is connected and contains neither loops nor parallel edges, and that it is already embedded in the plane without crossings (and with all the sources/sinks or terminals lying on the outer face, if needed).

Furthermore, as in~\cite{refBentzDAM09} (where all the necessary details are provided), we assume without loss of generality that all terminals are distinct and that the input graph is 2-vertex-connected (which can easily be achieved in linear time by doubling all edge weights, then obtaining a 2-edge-connected graph, and finally replacing any articulation vertex by a cycle). This means, in particular, that the boundary of the outer face is a simple cycle.

\section{A reformulation using clusterings}\label{sectClustering}

The starting point of our FPT algorithm is basically the same as in~\cite{refBentzIPEC12}: in any connected graph (planar or not), removing the edges of any optimal multicut yields several connected components, each of them containing at least one source or sink. The sources/sinks belonging to a same connected component define a \emph{cluster}. We shall call such a set of clusters a \emph{clustering} (a clustering is thus a partition of the sources and sinks), and we shall say that the considered solution \emph{induces} these clusters (and the connected components containing them), or equivalently this clustering.

Hence, finding an optimal multicut is equivalent to finding a set of edges of minimum total weight that isolates all the clusters of the clustering induced by this optimal multicut. By definition, this clustering is such that no cluster contains both $s_i$ and $s'_i$ for each $i$. In practice (i.e., from an algorithmic point of view), since we do not know this clustering as long as we do not know the optimal multicut itself, we need to enumerate all possible clusterings in order to ensure that the one induced by the optimal solution we are looking for will be considered as well.

The number of possible clusterings only depends on the number $k$ of source-sink pairs: in other words, it is FPT with respect to $k$ (we shall give more details later). This immediately implies that \textsc{MinMC} can be reduced in FPT time (with respect to $k$) to the following problem, called the \emph{minimum multi-cluster cut problem} (\textsc{MinMCC}): given $k'$ sets (or clusters) of terminals $\mathcal{T}_1, \mathcal{T}_2, \dots, \mathcal{T}_{k'}$ in an edge-weighted graph $G$, find a set of edges of minimum total weight in $G$, in such a way that removing these edges leaves no path between any vertex in $\mathcal{T}_i$ and any vertex in $\mathcal{T}_j$, for any $i \neq j$. When $\vert \mathcal{T}_i \vert = 1$ for each $i$, \textsc{MinMCC} simply turns into \textsc{MinMTC}. Also note that \textsc{MinMCC} is actually a special case of \textsc{MinMC}, in which the source-sink pairs are $(u,v)$, for each $i \neq j$ and each $u \in \mathcal{T}_i$ and $v \in \mathcal{T}_j$.

In the remainder of this paper, we shall focus on solving \textsc{MinMCC} in FPT time (with respect to $k'$), and this will immediately enable us to solve \textsc{MinMC} in FPT time (with respect to $k$) as well. Indeed, it is not hard to see that the above-mentioned reduction from \textsc{MinMC} to \textsc{MinMCC} is actually an FPT-reduction, as we have $k' \leq 2k$. One can even show a sharper bound on $k'$: we have actually $k' \leq k+1$, and this bound is \emph{tight} (to see this, simply consider a chain with $2k$ vertices, in the order $s_1,s'_1,s_2,s'_2,\dots,s_k,s'_k$, where the only optimal multicut induces $k+1$ clusters). This will not have any significant impact on the asymptotic running time of our algorithm, but we give a short proof of this fact anyway, for the sake of completeness.

Assume by contradiction that, in a given instance of \textsc{MinMC}, there exists an optimal multicut inducing at least $k+2$ connected components. Any edge of this optimal multicut lies between two of these connected components, and there must exist an $i$ such that one component contains $s_i$ and the other $s'_i$ (otherwise, this edge would be useless). Hence, when adding to the (at least) $k+2$ connected components induced by such an optimal multicut the edges lying between the two connected components containing $s_1$ and $s'_1$ (if there exist such edges), we obtain at least $k+1$ connected components. We do the same for $s_2$ and $s'_2$ (adding edges of the optimal multicut only if they have not already been added so far), and then for each $s_i$ and $s'_i$ for $i$ from $3$ to $k$. In the end, we have \emph{all} the edges from the initial graph, and we have reduced by \emph{at most} $k$ the number of connected components. Therefore, there remain at least two connected components, which contradicts the fact that the initial graph was connected.

\section{A characterization using planar duality}\label{sectPlanarDuality}

Let us now consider planar duality. It is well-known that, to any planar graph $G$ embedded in the plane, one can associate a \emph{dual} planar graph $G^*$. More precisely, to any face in $G$ corresponds a vertex (called a \emph{dual} vertex) in $G^*$, and any (\emph{dual}) edge between two dual vertices corresponds to the edge (or one of the edges, if there are more than one) shared by the corresponding faces in $G$. This also holds for the outer face of $G$. Furthermore, since $G$ is 2-vertex-connected, there exists no edge belonging to only one face.

Given an optimal solution $S$ for a \textsc{MinMCC} instance in a planar graph $G$, we shall denote by $S^*$ the set of dual edges corresponding to $S$ in the dual graph $G^*$, and, for each $i$, by $V_i$ the set of vertices of the $i$th connected component induced by $S$, and by $S_i \subseteq S$ the set of edges of $G$ having exactly one endpoint in $V_i$. It is well-known that $S^*_i$, the set of dual edges corresponding to each $S_i$, is a set of (non necessarily simple) cycles in $G^*$.

Moreover, if we look at the embedding of $G^*$ as a set of curves in the plane (which intersect at the dual vertices), then each $S^*_i$ is represented by a set of closed curves, denoted by $C^*_i$. Each $S^*_i$ is actually composed of one or several simple cycles $\{S^{*1}_i, S^{*2}_i, \dots\}$, i.e., each $C^*_i$ is composed of one or several simple closed curves $\{C^{*1}_i, C^{*2}_i, \dots\}$. By the \emph{Jordan curve theorem}, each such simple closed curve divides the plane into an interior region and an exterior region (a \emph{region} being a set of points such that any two of these points can be linked by a curve without crossing any closed curve): intuitively, the interior region is the region of the plane that \emph{is enclosed by} (or that \emph{lies inside}) this closed curve. The $C^*_i$'s thus partition the plane into several regions. Among all these regions, there is one and only one that is unbounded (and one and only one $C^*_i$ is \emph{associated with it}, i.e., contains all the curves adjacent to this region). Actually, the $C^*_i$'s may be seen as defining the boundary of the $V_i$'s, that form a partition of the vertex set of $G$, and hence the interior regions of any two distinct simple closed curves that compose them cannot overlap (except if one lies inside the other).

The following lemma summarizes well-known facts for planar \textsc{MinMCC}:

\begin{Lemma}[\cite{refBentzIPEC12,refECDV17,refDahlhaus94}]\label{lemma1:folklore}
Given a \textsc{MinMCC} instance in a planar graph $G$, any optimal solution $S$ for this instance satisfies the following properties:
\begin{itemize}
\item[$(1)$] for each $C^*_i=\{C^{*1}_i, C^{*2}_i, \dots\}$ and for any $j_1 \neq j_2$, either the interior regions of $C^{*j_1}_i$ and $C^{*j_2}_i$ are disjoint, or one lies inside the other,
\item[$(2)$] each $C^*_i=\{C^{*1}_i, C^{*2}_i, \dots\}$, except the one associated with the unbounded region, contains one simple closed curve, say $C^{*1}_i$, such that all the vertices of $V_i$ lie inside $C^{*1}_i$, but not inside $C^{*j}_i$ for any $j \geq 2$,
\item[$(3)$] for each $C^*_i$, except the one associated with the unbounded region, and for each $j \geq 2$, $C^{*j}_i$ lies inside $C^{*1}_i$, and, for each $j_1 \geq 2$ and $j_2 \geq 2$ with $j_1 \neq j_2$, the interior regions of $C^{*j_1}_i$ and $C^{*j_2}_i$ are disjoint.
\end{itemize}
\end{Lemma}
\begin{proof}
Each of these three properties has been more or less explicitly proved or used for solving \textsc{MinM(T)C} or \textsc{MinMCC} in planar graphs in~\cite{refBentzIPEC12,refECDV17,refDahlhaus94}. Property $(1)$ is clear from~\cite{refDahlhaus94}, and we will justify the other two briefly.

Concerning Property $(2)$, it was noticed in~\cite{refDahlhaus94} for \textsc{MinMTC} (and it can easily be extended to \textsc{MinMCC}) that, except for the $C^*_i$ associated with the unbounded region, each $C^*_i$ must enclose a region containing the terminals in $\mathcal{T}_i$. By definition, this region lies inside either the interior region or the exterior region associated with each $C^{*j}_i$. Obviously, it cannot lie inside several of these interior regions: indeed, from Property $(1)$, for any two of these interior regions, either they are disjoint (and thus it is clearly not possible), or one lies inside the other (and thus this other one is useless). Therefore, for each $i$, the region containing the terminals in $\mathcal{T}_i$ lies inside one of these interior regions (the one associated with $C^{*1}_i$) and outside all the other interior regions (or, equivalently, inside all the other exterior regions).

Concerning Property $(3)$, it is sufficient to notice that, from Property $(2)$, the only way for the interior region of $C^{*1}_i$ to intersect the exterior region of $C^{*j}_i$ for any $j \geq 2$ is to have $C^{*j}_i$ lying inside $C^{*1}_i$ for any $j \geq 2$. Moreover, this also implies that for any $C^{*j_1}_i$ and $C^{*j_2}_i$ with $j_1 \geq 2$, $j_2 \geq 2$ and $j_1 \neq j_2$, the interior region of one cannot lie inside the interior region of the other, and hence, from Property $(1)$, they must be disjoint.
\end{proof}

For the special case considered here, we can prove the following lemma:

\begin{Lemma}\label{lemma2:properties}
Given a \textsc{MinMCC} instance in a planar graph $G$ where all the terminals lie on the outer face, any optimal solution $S$ is such that each $S^{*j}_i$ contains the dual vertex corresponding to the outer face of $G$, and, for each $i$, any two $S^{*j}_i$'s have only this vertex in common.
\end{Lemma}
\begin{proof}
We begin by proving the first part of the statement. If some $S^{*j}_i$ did not contain the dual vertex corresponding to the outer face of $G$, then the interior region of $C^{*j}_i$ would not enclose any terminal (as any terminal lies on the outer face of $G$), and hence it would be useless in an optimal solution.

Moreover, it is easy to see that, for each $i$, any two $S^{*j}_i$'s have at most one vertex in common, even in the case where the terminals can lie anywhere in the planar graph (and hence, in our special case, they have exactly one vertex in common). Indeed, if two $S^{*j}_i$'s had two or more vertices in common, then they could not belong to the boundary of a single region.
\end{proof}

Lemma \ref{lemma2:properties} implies, in particular, that, if the input graph is planar and any terminal lies on the outer face, then \emph{any} $C^*_i$, including the one associated with the unbounded region, actually consists of a \emph{single} closed curve (that may not be simple). Let us assume without loss of generality that the closed curve associated with the unbounded region is $C^*_1$. From Lemma \ref{lemma1:folklore}, we call a cluster $\mathcal{T}_i$ with $i \geq 2$ a \emph{top cluster} if there is no $j \geq 2$ with $j \neq i$ such that $C^{*1}_i$ lies inside $C^{*1}_j$. (Besides, $\mathcal{T}_1$ will be referred to as a top cluster as well.)

This notion can be interpreted in the initial graph $G$ as well. For each $i$ and each $j$, let $S^j_i$ be the set of edges in $G$ associated with $S^{*j}_i$. Removing from $G$ the edges of any $S^1_i$ for $i \geq 2$ yields two connected components: one that contains the vertices in $\mathcal{T}_i$ (and possibly vertices from other clusters), and one that does not. We shall denote the former one by $V'_i$: we have $V_i \subseteq V'_i$ for each $i$. Then, $\mathcal{T}_i$ with $i \geq 2$ is a top cluster if $V'_i$ is not contained in any $V'_j$ for $j \geq 2$ and $j \neq i$. In other words, the vertices of any $V'_j$ with $j \geq 2$ such that $\mathcal{T}_j$ is not a top cluster are included in some $V'_i$ with $i \geq 2$ and $i \neq j$. As such, any $V'_i$ such that $\mathcal{T}_i$ is a top cluster and $i \geq 2$ can be viewed as a maximal inclusion-wise connected component among the $V'_j$'s.

However, this notion is still not strong enough to state our main result. We refine it as follows. Take \emph{any} top cluster except $\mathcal{T}_1$ (say, $\mathcal{T}_2$), and define it as a \emph{good} top cluster. Then, we define the other good top clusters iteratively: any top cluster $\mathcal{T}_i$ such that $S^{*1}_i$ with $i \geq 2$ has at least one edge in common with $S^{*1}_j$ for some good top cluster $\mathcal{T}_j$ with $j \geq 2$ will also be defined as a good top cluster. (Besides, $\mathcal{T}_1$ will be defined as a good top cluster as well.)

All the previous notions are illustrated in Figure \ref{fig1}, where the terminals are the small black rectangles, while the other vertices are the small black circles. The clusters are numbered from 1 to 9, and any terminal is labeled by the number of the cluster it belongs to. The dual vertices and edges associated with the optimal solution drawn in Figure \ref{fig1} are respectively the small grey diamonds and the grey dashed lines. The top clusters are numbered 1, 2, 5, 7 and 8, and the good top ones are numbered 1, 2 and 5 (another possible choice would be the ones numbered 1, 7 and 8). Moreover, the big grey diamond is the dual vertex associated with the outer face, and the three $S^{*j}_2$'s are indicated on the associated dual edges, as well as some other $S^{*j}_i$'s.

\begin{figure}%[htbp]
\centerline{\includegraphics[scale=0.8]{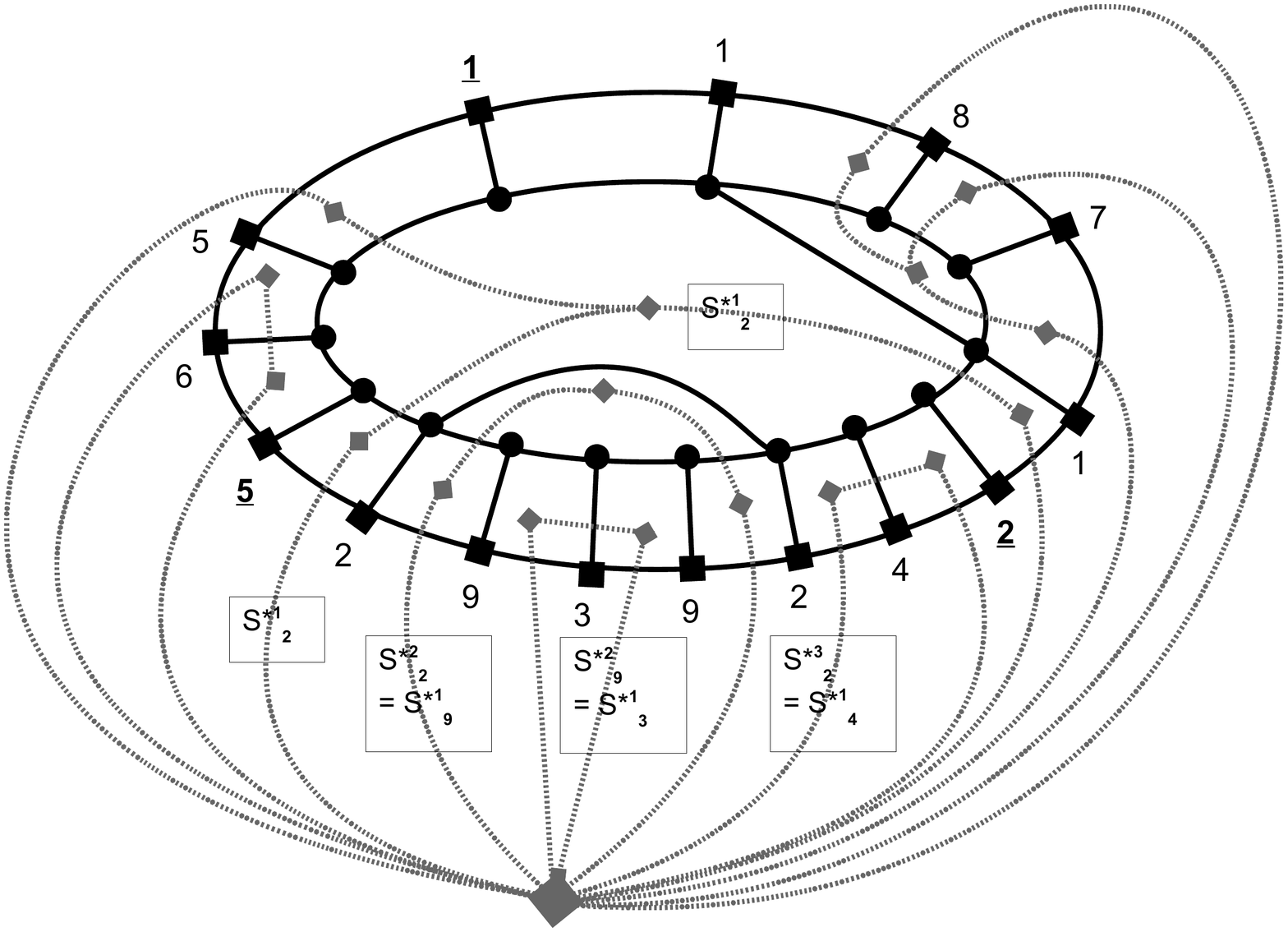}} %[scale=0.30]
\caption{A planar \textsc{MinMCC} instance and the associated optimal solution, whose edges have weight 1 (while all the other edges have large weights).}
\label{fig1}
\end{figure}

Our main result is the following lemma:

\begin{Lemma}\label{lemma3:main}
Assume we are given a \textsc{MinMCC} instance $I$ in a planar graph $G$ where any terminal lies on the outer face, and let $S$ be an optimal solution for $I$. Then, $\bigcup_{i:i\text{ is a good top cluster}} S^1_i$ is an optimal solution for the \textsc{MinMTC} instance $I'$ obtained as follows: $(i)$ the input graph $G'$ is the graph $G$ without the edges in $S \setminus \left(\bigcup_{i:i\text{ is a good top cluster}} S^1_i\right)$, and with one terminal for each good top cluster, and $(ii)$ all the terminals lie on the outer face.
\end{Lemma}
\begin{proof}
Consider the planar graph $G'$ as defined above, and assume that the embedding of $G'$ is computed by taking the one of $G$ and then simply removing the edges in $S \setminus \left(\bigcup_{i:i\text{ is a good top cluster}} S^1_i\right)$. Clearly, if, for each good top cluster $\mathcal{T}_i$, we add in $G'$ a new vertex (called a \emph{cluster vertex}) linked by an edge (called a \emph{cluster edge}) having a sufficiently large weight to each terminal of $\mathcal{T}_i$ (assume for now that it can be done in such a way that $(ii)$ holds), then, by the definitions of $S$ and $G'$, removing the edges in $\bigcup_{i:i\text{ is a good top cluster}} S^1_i$ leaves no path between any two cluster vertices. In other words, $\bigcup_{i:i\text{ is a good top cluster}} S^1_i$ is actually a feasible solution to $I'$.

Moreover, any edge in $\bigcup_{i:i\text{ is a good top cluster}} S^1_i$ lies between two connected components associated with two good top clusters. Indeed, on the one hand, by definition of a good top cluster, such an edge cannot belong to $S^{*1}_i$ (and hence, from Property $(3)$ in Lemma \ref{lemma1:folklore}, to $S^*_i$) for some top cluster $\mathcal{T}_i$ which is not good. On the other hand, from Property $(3)$ in Lemma \ref{lemma1:folklore} and the definition of a top cluster, a connected component associated with a non top cluster lies inside a closed curve of the form $C^{*j}_i$ for some $j \geq 2$, where $i \geq 2$ is such that $\mathcal{T}_i$ is a top cluster. From Lemma \ref{lemma2:properties}, the set of edges $S^{*j}_i$ associated with such a closed curve has one and only one vertex in common with any other $S^{*h}_i$, and hence, in particular, $S^{*j}_i$ shares no edge with $S^{*1}_i$.

This implies that in $G'$ there is no path from any terminal of any good top cluster to any other terminal, except for the terminals of any other good top cluster. Therefore, replacing $\bigcup_{i:i\text{ is a good top cluster}} S^1_i$ in $G'$ by any feasible solution to $I'$, whose total weight does not exceed the one of $\bigcup_{i:i\text{ is a good top cluster}} S^1_i$ and that contains no cluster edge, yields another feasible solution to $I$ in $G$ which is at least as good as $S$. As $S$ is an optimal solution to $I$, and as any optimal solution to $I'$ contains no cluster edge (their common weight being too large), this implies in particular that $\bigcup_{i:i\text{ is a good top cluster}} S^1_i$ is an optimal solution to $I'$.

It remains to prove the last part of the lemma: namely, let us prove that we can add the cluster vertices and edges in such a way that $(ii)$ holds.

To do this, we shall proceed in a way similar to the one described in~\cite{refBentzDAM09}: the claim that the cluster vertices can then be assumed to lie on the outer face will simply come from the fact that, unlike in~\cite{refBentzDAM09}, there is not a cluster vertex associated with each cluster, but only with each good top one. To describe our way of achieving this, we shall need some additional definitions.

%The boundary of the outer face of $G$ can be viewed as a simple closed curve, whose interior region encloses the embedding of $G$. The exterior region for this simple closed curve will be called \emph{the exterior region of $G$}.

Clearly, from the definition of a top cluster, all the terminals of such a cluster, \emph{except $\mathcal{T}_1$}, are consecutive on the outer face (among all the terminals of top clusters), as otherwise $C^{*1}_i$ would lie inside $C^{*1}_j$ for two distinct top clusters $\mathcal{T}_i$ and $\mathcal{T}_j$ with $i \geq 2$ and $j \geq 2$, which would be a contradiction. Therefore, if we go through the outer face of the graph $G$ clockwise (which can be done in a well-defined way, as $G$ is 2-vertex-connected), then, for each top cluster $\mathcal{T}_i$ of $S$ with $i \geq 2$, there is a ``first'' terminal of this cluster that is encountered when doing so while staying inside $C^{*1}_i$. In other words, each such top cluster $\mathcal{T}_i$ has a unique terminal (which we shall call the \emph{first} terminal of $\mathcal{T}_i$) from which we can encounter every other terminal of $\mathcal{T}_i$ by going through the outer face of the graph $G$ clockwise and without leaving the interior region of $C^{*1}_i$. From this first terminal, we can then define a unique ordering of the other terminals of $\mathcal{T}_i$, which is simply the order in which they are encountered while going through this outer face clockwise.

We can define the first (and last) terminal of $\mathcal{T}_1$ in a similar way, i.e., as the unique terminal of $\mathcal{T}_1$ from which we can encounter every other terminal of $\mathcal{T}_1$ by going through the outer face of the graph $G$ clockwise and without \emph{entering} the interior region of $C^{*1}_i$ for any good top cluster $\mathcal{T}_i$ with $i \geq 2$. Such a first terminal exists, as otherwise it would mean that, among the terminals in the good top clusters, the terminals in $\mathcal{T}_1$ are not consecutive on the outer face. In other words, it would mean that, for any choice of a first terminal of $\mathcal{T}_1$, there is a good top cluster $\mathcal{T}_i$ with $i \geq 2$ (resp. another good top cluster $\mathcal{T}_j$ with $j \geq 2$ and $j \neq i$) whose terminals are encountered while going clockwise (resp. counterclockwise) from the first terminal of $\mathcal{T}_1$ to its last one on the outer face. However, $S^{*1}_i$ and $S^{*1}_j$ for such $i$ and $j$ could not share an edge (as otherwise either the first or the last terminal of $\mathcal{T}_1$ would lie inside a closed curve belonging to $S$, contradicting the definition of $\mathcal{T}_1$), which would contradict the fact that $\mathcal{T}_i$ and $\mathcal{T}_j$ are good.

Hence, the point of introducing the notion of good top clusters is to extend the consecutiveness property associated with top clusters, \emph{even} when considering the cluster $\mathcal{T}_1$. In other words, all the terminals of \emph{any} good top cluster are consecutive on the outer face among \emph{all} the terminals of good top clusters. The notion of first vertices is illustrated in Figure~\ref{fig1}, where the first vertex of each of the three good top clusters (numbered 1, 2 and 5) in the optimal solution to the considered instance is indicated as follows: the number of the corresponding cluster is underlined and written in bold.

%In practice, in order to construct this planar \textsc{MinMTC} instance, we draw a curve (called a \emph{cluster curve}) from the first terminal of $\mathcal{T}_i$ to its last one, for each good top cluster $\mathcal{T}_i$ of $S$. These curves can easily be drawn in such a way that they do not intersect, and the curve corresponding to each $\mathcal{T}_i$ must be \emph{homotopic}, with respect to the boundary of the outer face of $G$, to the chain $\mu_i$ that goes clockwise from the first terminal of $\mathcal{T}_i$ to its last one, and uses only vertices lying on the boundary of this outer face. As in~\cite{refBentzDAM09}, being homotopic means that it can be continuously transformed into $\mu_i$ without being blocked by the boundary of this outer face while doing so (see also ???). Then, we let each cluster vertex lie on the associated cluster curve, and add the cluster edges, in such a way that they do not intersect and all lie inside the region bounded by the cluster curve and $\mu_i$.

Now we can proceed almost as in~\cite{refBentzDAM09}. For each good top cluster $\mathcal{T}_i$ of $S$ that contains at least two terminals (otherwise, there is nothing to do), we draw a curve (called a \emph{cluster curve}) from the first terminal of $\mathcal{T}_i$ to its last one. Thanks to the consecutiveness property associated with good top clusters, these curves can easily be drawn in such a way that no two of them intersect, and the curve corresponding to each $\mathcal{T}_i$ must be \emph{homotopic}, with respect to the boundary of the outer face of $G$, to the chain $\mu_i$ that goes clockwise from the first terminal of $\mathcal{T}_i$ to its last one, and uses only vertices lying on the boundary of this outer face. As in~\cite{refBentzDAM09}, being homotopic means that it can be continuously transformed into $\mu_i$ without being blocked by the boundary of this outer face while doing so (see also \cite{refECDVE10}).

Then, we let each cluster vertex lie on the associated cluster curve, and add the cluster edges, in such a way that they do not intersect and all lie inside the region bounded by the cluster curve and $\mu_i$. This allows us to conclude that, after adding the cluster vertices and edges as above, all the cluster vertices do lie on the outer face, which ends the proof.
\end{proof}

%L'hypothèse "that induces the maximum number of clusters among all the optimal solutions for $I$" est ici inutile (cf section suivante) !!!

% *** TO DO ***
% soumettre à Algorithmica/arXiv + inclure dans HDR
% relire ==> OK
% figure ==> OK
% écrire algo formellement en entier ==> OK
% traiter le cas du top cluster associé à la "unbounded region" ==> OK ???
% reprendre Section 4 (dont corollaire lemme 3) ==> OK ???
% runtime détaillé ==> OK
% conclusion ==> marche que si chaque cluster induit une CC !!! + partial multicut ==> OK
% ref partial multicut et ECDV & Erickson, ou non ? (En pratique (çàd pour l'algo), on trace la cluster curve et tout le tralala...) ==> OK

\section{Crafting the algorithm}\label{sectAlgo}

We now focus on using the results from the two previous sections to come up with an algorithm solving \textsc{MinMC} in time FPT with respect to $k$.

Let $I$ be an instance of \textsc{MinMC} in a planar graph where all sources and sinks lie on the outer face, and let $S$ be an optimal multicut for $I$ ($S$ exists, even if we do not know it yet explicitly).

First of all, we can ``guess'' the clustering associated with $S$ by enumerating \emph{all} the possible clusterings containing at most $k+1$ clusters, and this can be done in time FPT with respect to $k$ (see Section \ref{sectClustering}).

Then, we have to know the structure of the clusters in $S$, i.e., in particular, which clusters are the top ones, and which are the good top ones (see Section \ref{sectPlanarDuality}): we can ``guess'' this structure by using another enumeration (which, again, can be done in time FPT with respect to $k$).

Moreover, recall from the proof of Lemma \ref{lemma3:main} in Section \ref{sectPlanarDuality} that, if we go through the outer face of the input graph $G$ clockwise, then, for each good top cluster $\mathcal{T}_i$ of $S$ with $i \geq 2$, there is a first terminal lying inside $C^{*1}_i$ that is encountered while doing so (a similar notion of a first terminal holds for $\mathcal{T}_1$ as well). As mentioned in this proof, we will need to know this terminal in order to ensure that, in the planar \textsc{MinMTC} instance that we will construct, all the terminals will lie on the outer face.

Again, for each good top cluster $\mathcal{T}_i$, we can ``guess'' such a terminal by enumerating all the possibilities (i.e., trying the $\vert \mathcal{T}_i \vert$ terminals of $\mathcal{T}_i$ one by one). Once we know the first terminal of each such cluster, we can construct a planar \textsc{MinMTC} instance as explained in the proof of Lemma \ref{lemma3:main}.

By solving the above-defined planar \textsc{MinMTC} instance (where all the terminals lie on the outer face), we obtain from Lemma \ref{lemma3:main} a set of edges $S'$ that we can use to replace $\bigcup_{i:i\text{ is a good top cluster}} S^1_i$ in $S$. However, the main drawback of Lemma \ref{lemma3:main} is that it considers a \textsc{MinMTC} instance defined on a graph $G'$ that is obtained from the input graph $G$, but that we do not know explicitly (as it would require to already know some part of an optimal solution). We now show how we can overcome this issue, by considering a \emph{particular} optimal solution to the \textsc{MinMCC} instance we wish to solve:

\begin{Corollary}\label{myCorollary}
Assume we are given a \textsc{MinMC} instance $I$ in a planar graph $G$ where all the sources and sinks lie on the outer face, and consider an optimal solution $S$ for $I$ that induces the maximum number of clusters. Then, $\bigcup_{i:i\text{ is a good top cluster}} S^1_i$ is an optimal solution for the \textsc{MinMTC} instance $I'$ obtained as follows: $(i)$ the input graph is $G$, with one additional terminal for each good top cluster, and $(ii)$ all the terminals lie on the outer face.
\end{Corollary}
\begin{proof}
From Lemma \ref{lemma3:main}, we just have to prove that, in this case (i.e., when we consider an optimal solution $S$ to $I$ inducing the maximum number of connected components), we do not have to know $S \setminus \left(\bigcup_{i:i\text{ is a good top cluster}} S^1_i\right)$ explicitly in order to define $I'$. In other words, that any optimal solution to $I'$ does not interact with $S \setminus \left(\bigcup_{i:i\text{ is a good top cluster}} S^1_i\right)$, i.e., does not share any edge with it. If these sets of edges did intersect, then from the proof of Lemma \ref{lemma3:main} this would yield another optimal solution for $I$, that would induce more connected components than $S$ does, contradicting the choice of $S$.
\end{proof}

~\\
Since the first step of our algorithm is to enumerate all the possible clusterings containing at most $k+1$ clusters, we will in particular consider the one associated with such an optimal solution $S$.

By solving the above-defined planar \textsc{MinMTC} instance with $n$ vertices (where all the terminals lie on the outer face), which can be done in time $O(k^3 n + k^2 n \log n)$ thanks to the algorithm proposed in~\cite{refChen04}, we obtain from this corollary a set of edges $S'$ that we can use to replace $\bigcup_{i:i\text{ is a good top cluster}} S^1_i$ in $S$, and, furthermore, that does not intersect $S \setminus \left(\bigcup_{i:i\text{ is a good top cluster}} S^1_i\right)$. Moreover, the graph of this instance is obtained from $G$ simply by adding cluster vertices and edges, and hence we do not have to know $S \setminus \left(\bigcup_{i:i\text{ is a good top cluster}} S^1_i\right)$ explicitly. After removing $S'$ from $G$, we obtain two or more connected components (and hence knowing $S \setminus \left(\bigcup_{i:i\text{ is a good top cluster}} S^1_i\right)$ explicitly or not is irrelevant).

These components, in turn, define smaller \textsc{MinMCC} instances, which can then be solved recursively by using the same strategy as above (without the first step, where we guessed the clustering), in a divide-and-conquer way. Each time an instance is solved, the number of connected components increases by at least one: as $S$ contains at most $k+1$ such components, there is a total of at most $k$ instances to be solved.

Putting all together, we obtain the following algorithm $A_1$, which makes a call to another algorithm, that will be detailed after $A_1$:\\

\underline{\textbf{Algorithm} $A_1$}

\textbf{Input:} A connected planar graph $G$ with $n$ vertices, and a set of $k$ source-sink pairs $(s_1,s'_1), \dots, (s_k,s'_k)$ lying on the outer face of $G$.

\textbf{Output:} An optimal multicut for the input graph.

\begin{itemize}
\item For each clustering containing $2k$ terminals and $\leq k+1$ clusters do:
\begin{itemize}
\item Build the associated planar \textsc{MinMCC} instance, where any terminal lies on the outer face of $G$, and the clusters are $\mathcal{T}_1, \mathcal{T}_2, \dots$,
\item Run Algorithm $A_2(G,\{\mathcal{T}_1, \mathcal{T}_2, \dots\})$, and store its output.
\end{itemize}
\item Output the best feasible solution found.
\end{itemize}

Observe that the input of Algorithm $A_1$ is a \textsc{MinMC} instance, while the input of Algorithm $A_2$ will be a \textsc{MinMCC} instance. If a given clustering is induced by an optimal multicut but does not contain the maximum number of clusters, then the solution computed by $A_1$ \emph{for this cluster} may not be optimal (which simply means that we need to consider another clustering).

Moreover, in Algorithm $A_1$, each call to Algorithm $A_2$ is actually the first call of a series of recursive calls. In other words, Algorithm $A_2$ is a recursive algorithm, that can be described as follows:\\

\underline{\textbf{Algorithm} $A_2$}

\textbf{Input:} A connected planar graph $G$ with $n$ vertices, and a set $\{\mathcal{T}_1, \mathcal{T}_2, \dots\}$ of clusters of terminals, all lying on the outer face of $G$.

\textbf{Output:} An optimal multi-cluster cut for the input graph.

\begin{itemize}
\item For each possible choice of good top clusters among all the clusters of the clustering $\{\mathcal{T}_1, \mathcal{T}_2, \dots\}$, and for each possible choice of first terminals for all these good top clusters, do:
\begin{itemize}
\item Construct and solve the associated planar \textsc{MinMTC} instance, where the terminals, that are the cluster vertices associated with the good top clusters, all lie on the outer face (see Corollary \ref{myCorollary}),
\item Remove from the input graph the edges of the optimal solution computed above, obtaining several connected components $G_1, G_2, \dots$, and then store these edges in the current solution,
\item For each of these connected components $G_i$ that contains the terminals of at least two clusters, run Algorithm $A_2(G_i, \mathcal{T}(G_i))$, where $\mathcal{T}(G_i)$ is the set of clusters whose terminals belong to $G_i$, and then add the associated output to the current solution.
\end{itemize}
\item Output the best feasible solution found.
\end{itemize}

Thanks to the above discussion, it should be clear that Algorithm $A_1$ is correct, and runs in time $O(f(k) n \log n)$, for some function $f(\cdot)$ to be specified, in graphs with $n$ vertices. This running time is actually obtained by multiplying the different factors associated with the successive steps:
\begin{enumerate}
\item Enumerating all the possible clusterings containing $2k$ terminals and at most $k+1$ clusters incurs a factor $O\left(\frac{(k+1)^{2k}}{(k+1)!}\right)$, as noted in~\cite{refDahlhaus94},
\item Enumerating all the possible good top clusters among the (at most $k+1$) clusters of a given clustering, and then all their possible first terminals (among $O(k)$), incurs a factor $O(k 2^{k+1}) = O(k 2^k)$,
\item Solving each planar \textsc{MinMTC} instance with all the terminals lying on the outer face can be done in time $O(k^2 (kn + n \log n))$,
\item Finally, there are at most $k$ such instances to solve.
\end{enumerate}

The overall running time is thus $O\left(k^4 2^k \frac{(k+1)^{2k}}{(k+1)!}(kn + n \log n)\right)$, i.e., it is nearly linear when $k = O(1)$. Hence, we have proved:

\begin{Th}\label{th:main}
In planar graphs where all sources and sinks lie on the outer face, \textsc{MinMC} is FPT with respect to the number $k$ of source-sink pairs.
\end{Th}

\section{Extensions and open problems}

In this paper, we have provided an FPT algorithm for \textsc{MinMC} parameterized by the number $k$ of source-sink pairs, in the case where the input graph is planar and all the sources and sinks lie on the outer face. This algorithm actually runs in $O(n \log n)$ time when $k = O(1)$, where $n$ is the number of vertices of the input graph. In~\cite{refBentzDAM09}, it was proved that the time for solving this problem can be improved to linear when $k=2$, but the proof cannot be generalized to greater values of $k$. Therefore, this set of results leaves as open the following question: does there exist a linear-time algorithm (i.e., running in time $O(n)$ for any $k = O(1)$) in this case?

Moreover, our FPT algorithm can easily be extended to a generalization of \textsc{MinMC}, called \emph{partial} \textsc{MinMC} (or \emph{$k$-multicut problem} \cite{refGNS06}), which asks to select a minimum-weight set of edges whose removal leaves no path from $s_i$ to $s'_i$, for at least a given number of source-sink pairs $(s_i,s'_i)$. Indeed, partial \textsc{MinMC} can be reduced in FPT time to \textsc{MinMC}, by ``guessing'' the subset of source-sink pairs between which there will remain no path in an optimal solution (there are $O(2^k)$ such possible subsets to enumerate).

Let us now consider \textsc{MinMCC}. On the one hand, it is easy to see that \textsc{MinMCC} is polynomial-time solvable in general graphs with two clusters (by reducing it to the minimum cut problem), but Dahlhaus et al. (that call it the \emph{colored multiterminal cut problem}) proved in~\cite{refDahlhaus94} that it is \textbf{NP}-hard in planar graphs, even with only four clusters (and they claimed that this remains true with only three clusters). On the other hand, when the input graph is planar and has all its terminals lying on the outer face, \textsc{MinMTC} (the special case of \textsc{MinMCC} where clusters have size 1) is polynomial-time solvable, even when $\arrowvert \mathcal{T} \arrowvert$ is part of the input \cite{refChen04}, and our FPT algorithm precisely solves \textsc{MinMCC} parameterized by the total number of terminals in such a graph (or, equivalently, parameterized both by the number of clusters \emph{and} by the maximum number of terminals per cluster).

However, when the number of clusters is part of the input, we do not even know the complexity of \textsc{MinMCC} in such a graph. Observe that this question remains open even if there are $O(1)$ terminals in each cluster.

When the number of clusters is viewed as a parameter, one may hope that our approach is able to solve the problem even when there is an arbitrary number of terminals in each cluster, as we only need to ``guess'' the first terminal of each good top cluster. Unfortunately, this is true only if each cluster induces exactly one connected component in any optimal solution (see Corollary \ref{myCorollary} and the discussion preceding its statement), and hence, when such a property does not hold, this question remains open as well.

% non : on a quand même besoin que TOUT cluster induise une seule CC dans toute sol opt !!!
% + partial multicut !!!

% open problems : linear-time algo si k=O(1) (si k>2) ? complexité du multi-cluster cut pb dans les graphes planaires si nb de clusters non fixé mais terminaux sur la face extérieure ?

%\section*{Acknowledgments}

%The author thanks the anonymous referees for their useful remarks and comments.

%\newpage

\end{document}